\documentclass[11pt]{article}

\usepackage{amsmath, amssymb, amsthm} 

\usepackage{algorithm, algorithmicx}
\usepackage{float}

\usepackage[utf8x]{inputenc} 

\allowdisplaybreaks
\expandafter\let\expandafter\oldproof\csname\string\proof\endcsname
\let\oldendproof\endproof
\renewenvironment{proof}[1][\proofname]{%
	\oldproof[\bf #1]%
}{\oldendproof}

\allowdisplaybreaks

\theoremstyle{plain}
\newtheorem{theorem}{Theorem}
\newtheorem{lemma}{Lemma}[section]

\newtheorem{proposition}[lemma]{Proposition}
\newtheorem{observation}[lemma]{Observation}
\newtheorem{corollary}[theorem]{Corollary}
\newtheorem{conjecture}[lemma]{Conjecture}
\newtheorem{problem}[lemma]{Problem}

\newtheorem{definition}[lemma]{Definition}

\newtheorem{fact}[lemma]{Fact}

\newcommand{\inj}{\mathrm{inj}}
\newcommand{\ind}{\mathrm{ind}}

\RequirePackage[normalem]{ulem} 
\RequirePackage{color}\definecolor{RED}{rgb}{1,0,0}\definecolor{BLUE}{rgb}{0,0,1} 

\newcommand{\Hdir}{\vec{H}}
\newcommand{\FHdir}{F_{\vec{H}}}

\newcommand{\didist}{\overrightarrow{\text{dist}}}

\usepackage{hyperref} 

\title{Counting Subgraphs in Degenerate Graphs}
\author{Suman K. Bera\thanks{University of California, Santa Cruz. Supported by NSF TRIPODS grant CCF-1740850, NSF CCF-1813165, CCF-1909790, CCF-2023495, and
		ARO Award W911NF1910294. Emails: \{sbera,sesh\}@ucsc.edu.}
	\and Lior Gishboliner\footnotemark[2]
	\and Yevgeny Levanzov\thanks{School of Mathematics, Tel Aviv University, Tel Aviv 69978, Israel. Supported in part by ISF Grant 1028/16 and ERC Starting Grant 633509. Emails: \{liorgis1,yevgenyl,asafico\}@mail.tau.ac.il.}
	\and C. Seshadhri\footnotemark[1]
	\and Asaf Shapira\footnotemark[2]}

\begin{document}
	
	\maketitle
	
	\begin{abstract}
		We consider the problem of counting the number of copies of a fixed graph $H$ within an input graph $G$. This is one of the most well-studied algorithmic graph problems, with many theoretical and practical applications. We focus on solving this problem when the input $G$ has {\em bounded degeneracy}. This is a rich family of graphs, containing all graphs without a
		fixed minor (e.g. planar graphs), as well as graphs generated
		by various random processes (e.g. preferential attachment graphs). We say
		that $H$ is {\em easy} if there is a linear-time algorithm for counting the
		number of copies of $H$ in an input $G$ of bounded degeneracy.
		A seminal result of Chiba and Nishizeki from '85 states that every $H$ on at
		most 4 vertices is easy. Bera, Pashanasangi, and Seshadhri recently extended
		this to all $H$ on 5 vertices, and further proved that for every $k > 5$ there is
		a $k$-vertex $H$ which is not easy. They left open the natural problem of characterizing all easy graphs $H$. 
		
		Bressan has recently introduced a framework for counting subgraphs in degenerate graphs, from which one can extract a sufficient condition for a graph $H$ to be  easy. Here we show that this sufficient condition is also necessary, thus fully
		answering the Bera--Pashanasangi--Seshadhri problem. 
		We further resolve two closely related problems; namely characterizing the graphs that are easy with respect to counting induced copies, and with respect to counting homomorphisms. 
		
		
	\end{abstract}
	
	\section{Introduction}\label{sec:intro}
	
	Subgraph counting refers to the algorithmic task of computing the number of copies (i.e., occurences) of a given graph $H$ in an input graph $G$. 
	Due to its fundamental nature, this problem has been studied extensively, both from a theoretical perspective and for practical applications. 
	In practice, subgraph counts are widely used to analyze real-world graphs, such as graphs representing telecommunication networks, biological structures and social interactions. Consequently, subgraph counts feature prominently in studies of biological \cite{HBPC,Przulj,PCJ} and sociological \cite{Burt,HL} networks, as well as in the network science literature in general \cite{BGL,KLYGCT,MSIKCA,SSPC,Tsourakakis,UBK}. 
	For example, \cite{MSIKCA} observed that networks coming from different areas of science (such as biochemistry, neurobiology, ecology, and engineering) have significantly different counts of small subgraphs.
	Such frequently occuring subgraphs are called {\em motifs}, and, quoting \cite{MSIKCA}, ``may uncover the basic building blocks of most networks".  
	Needless to say, some real-world graphs can be very large -- having billions of vertices -- thus making it all the more desirable to have fast subgraph counting algorithms.   
	
	In theoretical computer science, subgraph counting and detection\footnote{The $H$-detection problem is the problem of deciding whether an input graph contains a copy of $H$.} are fundamental and widely-studied problems. Much of the research focused on counting special kinds of graphs, such as cliques \cite{EG,IR,NP,V W}; cycles \cite{AYZ,DVW,IR}; paths and matchings (and other graphs with bounded pathwidth) \cite{BHKK,BKL}; graphs with bounded vertex-cover number \cite{CDM,C_Marx,KLL,VW-W}. Many of the algorithms use fast matrix multiplication \cite{AYZ,EG,IR,KLL,NP}. For example, the best known algorithm for counting $k$-cliques \cite{NP} runs in time 
	$n^{\omega k/3+O(1)}$, where $\omega < 2.373$ is the matrix multiplication constant \cite{V W 2}. On the negative side, $k$-clique counting is the canonical $\#\text{W}[1]$-hard problem, and it is thus unlikely that there exists an algorithm which solves this problem in time $f(k) \cdot n^{o(k)}$ (for any function $f$).   
	We refer the reader to \cite{BPS} for further references on both theoretical and practical aspects of subgraph counting. 
	
	Subgraph counts also play a fundamental role in extremal graph theory, where subgraph densities are the basic notion used for studying sequences of dense graphs \cite{Lovasz}. In particular, knowing approximate subgraph counts of small graphs inside a given graph $G$ allows one to decide whether $G$ is quasirandom \cite{CGW} or, more generally, whether $G$ consists of a bounded number of quasirandom pieces with prescribed edge densities \cite{LS}.
	
	Given the importance of the subgraph counting problem on the one hand, and its hardness in general graphs on the other, it is natural to consider special classes of graphs which admit faster counting algorithms, while also being rich enough to include many of the real-world graphs mentioned above. One prime example of such a family of graph classes is classes having bounded degeneracy. Recall that a graph $G$ is {\em $\kappa$-degenerate}\footnote{We note that degeneracy is closely related to another well-studied graph parameter, namely arboricity, which is the minimum number of forests into which the edge-set of a graph can be partitioned. It is well-known that the arboricity of a $\kappa$-degenerate graph is between $(\kappa+1)/2$ and $\kappa$.} if there is an ordering $v_1,\dots,v_n$ of the vertices of $G$ such that $v_i$ has at most $\kappa$ neighbours in $\{v_{i+1},\dots,v_n\}$ (for each $1 \leq i \leq n$). We say that a class of graphs has {\em bounded degeneracy} if there is an integer $\kappa$ such that all graphs in the class are $\kappa$-degenerate. 
	With a slight abuse of terminology, we will refer to graphs belonging to such classes as having bounded-degeneracy or being $O(1)$-degenerate.
	
	There are many examples of well-studied graph classes having bounded degeneracy. These include all minor-closed classes (including planar graphs and, more generally, graphs embeddable into a given surface), preferential attachment graphs \cite{BA}, and bounded expansion graphs \cite{NDeM}. 
	
	The first result on subgraph counting in bounded-degeneracy graphs is probably the classical result of Chiba and Nishizeki \cite{CN}, who showed that in $\kappa$-degenerate graphs, one can count $r$-cliques in time $O(n \kappa^{r-2})$ (for each $r \geq 3$), and $4$-cycles in time $O(n \kappa)$. 
	Bera, Pashanasangi and Seshadhri \cite{BPS} recently extended this result, by showing that the $H$-counting problem in bounded-degeneracy graphs can be solved in time $O(n)$ for every graph $H$ on at most $5$ vertices (here, the implicit constant in the big-$O$ notation depends on the degeneracy $\kappa$). They further showed that under a certain widely-believed hardness assumption in fine-grained complexity, the problem of counting $6$-cycles cannot be solved in linear time in bounded-degeneracy graphs. (They also proved a similar result for all longer cycles, with the exception of the cycle of length $8$.) 
	The hardness assumption used asserts that detecting a triangle in a (general) graph with $m$ edges requires significantly more than $O(m)$ time. 
	This assumption will also serve as the foundation for our complexity-theoretic hardness results, and we state it here as follows. We refer the reader to  \cite{A-VW}, where Conjecture \ref{conj:triangle_detection} was first formulated, for a detailed overview of the conjecture and its relations to many other computational problems. 
	
	\begin{conjecture}[\textsc{Triangle Detection Conjecture} \cite{A-VW}]\label{conj:triangle_detection}
		There exists $\gamma > 0$ such that in the word RAM model of $O(\log n)$ bits, any algorithm to decide whether an input graph with $n$ vertices and $m$ edges is triangle-free requires $\Omega(m^{1+\gamma})$ time in expectation. 
	\end{conjecture}  
	
	It is believed that the constant $\gamma$ in Conjecture \ref{conj:triangle_detection} could be as large as $1/3$. The reason for this is that the best known algorithm for triangle detection \cite{AYZ} runs in time 
	$O\big(\min \big( n^{\omega}, m^{2\omega/(\omega+1)} \big)\big)$, where $\omega$ is the matrix multiplication constant. If $\omega = 2$ (which would be optimal), then the running time of this algorithm is 
	$O\big(\min \big( n^2, m^{4/3} \big)\big)$.
	
	Having established both positive and negative results,
	Bera, Pashanasangi and Seshadhri \cite{BPS} asked whether one can characterize all graphs $H$ for which the $H$-counting problem can be solved in linear time in bounded-degeneracy graphs. 
	Before proceeding to our resolution of this question, let us recall some definitions
	and introduce some notation. 
	\begin{definition}\label{def:hom-1-1-hom-iso}
		Let $G,H$ be graphs. A map $\varphi: V(H) \rightarrow V(G)$ is
		\begin{enumerate}
			\item a {\em homomorphism} if $\{\varphi(u),\varphi(v)\} \in E(G)$ for every $\{u,v\} \in E(H)$;
			
			\item an {\em injective homomorphism} if it is a homomorphism and a one-to-one function;
			
			\item an {\em isomorphism} if it is an injective homomorphism and $\{\varphi(u),\varphi(v)\} \notin E(G)$ if $\{u,v\} \notin E(H)$.
		\end{enumerate} 
	\end{definition}	
	\noindent Let $G,H$ be graphs. Following \cite{Lovasz}, we denote by $\hom(H,G)$ the number of homomorphisms
	from $H$ to $G$, by $\inj(H,G)$ the number of injective homomorphisms (i.e., embeddings) from $H$ to $G$, and by $\ind(H,G)$ the number of isomorphisms from $H$ to (induced subgraphs of) $G$.  
	For a graph $H$, we denote by $\textsc{hom-cnt}_{H}$ the problem of computing $\hom(H,G)$ for a given input graph $G$. The problems $\textsc{inj-cnt}_{H}$ and $\textsc{ind-cnt}_{H}$ are defined analogously. 
	In what follows, we say that $\textsc{hom-cnt}_{H}/\textsc{inj-cnt}_{H}/\textsc{ind-cnt}_{H}$ is {\em easy} if it can be solved in time $f(\kappa,H) \cdot \tilde{O}(n)$ in $n$-vertex graphs of degeneracy $\kappa$ (for some function $f$); and otherwise we say that it is {\em hard}. We will usually avoid mentioning the function $f$, and just speak of running time $\tilde{O}(n)$, namely (nearly) linear time, with the implicit constant in the big-$\tilde{O}$ notation allowed to depend on $\kappa$ and $H$. 
	The following is the main open problem raised in \nolinebreak \cite{BPS}.
	\begin{problem}[\cite{BPS}]\label{prob:BPS}
		Characterize the graphs $H$ for which $\textsc{inj-cnt}_{H}$ is easy. 
	\end{problem}
	
	

	In this paper, we completely resolve Problem \ref{prob:BPS} by giving a very clean characterization of the graphs $H$ for which $\textsc{inj-cnt}_{H}$ is easy.
	We will also solve the related problems of characterizing the graphs $H$ for which $\textsc{hom-cnt}_{H}$ is easy and the graphs $H$ for which $\textsc{ind-cnt}_{H}$ is easy. It will turn out to be more convenient to first deal with the problem of obtaining a characterization for $\textsc{hom-cnt}_{H}$.
	This characterization, discussed in Section \ref{subsec:hom}, constitutes the main result of this paper. In Section \ref{subsec:inj} we describe how the solution
	of Problem \ref{prob:BPS} regarding $\textsc{inj-cnt}_{H}$ can be derived from our result regarding $\textsc{hom-cnt}_{H}$. Then, in Section \ref{subsec:ind}, we describe how a characterization for $\textsc{ind-cnt}_{H}$ can be derived from the one for $\textsc{inj-cnt}_{H}$. This approach -- of relating homomorphisms to copies and copies to induced copies -- was pioneered in \cite{CDM} and \cite{CM}, and is based on the framework developed in \cite{Lovasz}. 
	Finally, in Section \ref{subsec:general_graphs} we briefly consider subgraph-counting in {\em general} (i.e., not necessarily degenerate) graphs. By applying the methods used for proving Theorem \ref{thm:main}, we
	show that for a graph $H$, $\hom(H,G)$ can be computed in time $\tilde{O}(|V(G)| + |E(G)|)$ in general graphs if and only if $H$ is a forest (again, the ``only-if" direction assumes Conjecture \ref{conj:triangle_detection}).

	\subsection{$\alpha$-acyclic hypergraphs and Bressan's algorithm}
	Bressan \cite{Bressan} provided a dynamic programming algorithm for computing $\hom(H,G)$ in $O(1)$-degenerate graphs $G$.
	His main result is that computing $\hom(H,G)$ in $n$-vertex $O(1)$-degenerate graphs can be done in time\footnote{If one forgoes the requirement that the algorithm be deterministic, instead allowing 
	Las Vegas randomized algorithms (as is done for example in \cite{BPS}), then one can avoid the implicit polylogarithmic factor, obtaining an algorithm which runs in expected time $O(n^{\tau_1(H)})$. Indeed, the polylogarithmic factor arises from the need to search and update a dictionary with $O(n^{\tau_1(H)})$ entries. This can be done in time $O(1)$ by using perfect hashing \cite{FKS}, at the cost of needing expected time $O(n^{\tau_1(H)})$ to generate the hash table. Once generated, the algorithm proceeds deterministically.}
	$\tilde{O}(n^{\tau_1(H)})$, where $\tau_1(H)$ is a certain ``width parameter" called the {\em DAG treewidth} of $H$ (we use the notation $\tau_1$ to be consistent with \cite{Bressan}). 
	In our language, the special case $\tau_1(H) = 1$ of Bressan's result states that $\textsc{hom-cnt}_H$ is easy if $\tau_1(H) = 1$.
	As we will mainly focus on the case $\tau_1(H) = 1$, we shall not define $\tau_1(H)$ here, but shall only define what it means for a graph to satisfy $\tau_1(H) = 1$. To do so, we first need to recall the well-known notion 
	of hypergraph $\alpha$-acyclicity, which is closely related to the \nolinebreak parameter \nolinebreak $\tau_1$. 
	
	\begin{definition}[$\alpha$-acyclic hypergraph]\label{def:alpha-acyclic-hyper}
		A hypergraph $F$ is called \textup{$\alpha$-acyclic} if there exists a tree $T$ whose vertices are the hyperedges of $F$, such that the following condition is satisfied: for all $e_1,e_2,e \in E(F) = V(T)$, if $e$ is on the unique path in $T$ between $e_1$ and $e_2$, then $e_1 \cap e_2 \subseteq e$. 
	\end{definition}
	Hypergraph $\alpha$-acyclicity was introduced by Beeri, Fagin, Maier, Mendelzon, Ullman and Yannakakis \cite{BFMMUY} in the early 1980s in connection with relational database schemes, and has subsequently been widely studied. For further information, we refer the reader to 
	\cite{BFMY,B-B}.
	Note that in the special case that $F$ is a graph, $\alpha$-acyclicity is equivalent to being a forest.

	Next, one associates to each directed acyclic graph\footnote{We note that directed acyclic graphs arise naturally in the context of counting subgraphs in degenerate graphs. Indeed, many of the known counting algorithms \cite{CN,BPS,Bressan} begin by orienting the edges of the given $\kappa$-degenerate input graph $G$ according to some {\em degeneracy ordering}, thus obtaining an acyclic directed graph $\vec{G}$ in which all out-degrees are at most $\kappa$. Then the task becomes to compute $\hom(\vec{H},\vec{G})$ for every acyclic orientation $\vec{H}$ of $H$. Summing these directed homomorphism counts then gives \nolinebreak $\hom(H,G)$.} (or DAG, for short) $\vec{H}$ a hypergraph which captures the {\em reachability structure} of $\vec{H}$. 
	For a vertex $u$ in a directed graph $\vec{H}$, we denote by $R(u)$ the set of vertices that are \textit{reachable} from $u$ (namely, the set of all $v \in V(\vec{H})$ such that there is a directed path from $u$ to $v$). Note that $u \in R(u)$.
	 Recall that any DAG has at least one source (i.e. vertex of in-degree $0$). Given a DAG $\vec{H}$ with source vertices $u_1,\dots,u_r$, define an associated hypergraph $\FHdir$, as follows: the vertices of $\FHdir$ are the vertices of $\Hdir$, and the hyperedges of $\FHdir$ are the sets $R(u_i)$ for $1 \leq i \leq r$. Using this definition (and with a slight abuse of terminology), we can now define what it means for a DAG or an undirected graph to be $\alpha$-acyclic. 
	
	\begin{definition}[$\alpha$-acyclic graphs and digraphs]\label{def:alpha-acyclic-dir}
		 We say that a DAG $\Hdir$ is $\alpha$-acyclic if the hypergraph $\FHdir$ is $\alpha$-acyclic. We say that an undirected graph $H$ is $\alpha$-acyclic if every acyclic orientation $\Hdir$ of $H$ \nolinebreak is \nolinebreak $\alpha$-acyclic.  
	\end{definition}
	
	For an undirected graph $H$, having $\tau_1(H) = 1$ is precisely equivalent to being $\alpha$-acyclic (this is merely a restatement of the definition of $\tau_1$ for the case $\tau_1(H) = 1$ in a different language). As mentioned above, a special case of the main result of \cite{Bressan} implies that $\hom(H,G)$ can be computed in time $\tilde{O}(n)$ whenever $H$ is $\alpha$-acyclic. Our main result, described in the following section, states that this sufficient condition is also necessary, thus giving a complete characterization of the graphs $H$ for which $\textsc{hom-cnt}_H$ \nolinebreak is \nolinebreak easy.

	\subsection{Main result: counting homomorphisms in linear time}\label{subsec:hom}

Our main result in this paper is as follows. 


\begin{theorem}[Main result]\label{thm:main}
	Assuming Conjecture \ref{conj:triangle_detection}, $\textsc{hom-cnt}_{H}$ is hard whenever $H$ is not $\alpha$-acyclic. 
\end{theorem}

As mentioned above, Theorem \ref{thm:main} shows that the sufficient condition (for $\textsc{hom-cnt}_{H}$ being easy) supplied by Bressan's algorithm \cite{Bressan} is in fact necessary. Thus we obtain the following characterization:

\begin{corollary}\label{cor:main}
	$\textsc{hom-cnt}_{H}$ is easy if and only if $H$ is $\alpha$-acyclic. The ``only if" part is under Conjecture \ref{conj:triangle_detection}.
\end{corollary} 


Given Theorem \ref{thm:main}, it is natural to ask if there is a cleaner description of the $\alpha$-acyclic graphs.
Actually, another reason for seeking such a clean description is that in order to prove Theorem \ref{thm:main}, it would be desirable to know that graphs that are not $\alpha$-acyclic have certain easy-to-describe obstructions.
Luckily (and somewhat surprisingly), we have the following concise equivalent description of the $\alpha$-acyclic graphs. Throughout the paper, $C_k$ denotes the cycle of length $k$. For graphs $H,H_0$, we say that $H$ is {\em induced $H_0$-free} if $H$ contains no induced copy of $H_0$. 

\begin{theorem}\label{thm:alpha-acyclic-no-cycles}
	An undirected graph $H$ is $\alpha$-acyclic if and only if $H$ is induced $C_k$-free for every $k \geq 6$.
\end{theorem}

By combining Corollary \ref{cor:main} and Theorem \ref{thm:alpha-acyclic-no-cycles}, we immediately obtain the following very clean characterization of the graphs $H$ for which $\textsc{hom-cnt}_{H}$ is easy.
\begin{corollary}\label{cor:charac}
	$\textsc{hom-cnt}_{H}$ is easy if and only if $H$ is induced $C_k$-free for all $k \geq 6$. The ``only if" part is under Conjecture \ref{conj:triangle_detection}.
\end{corollary}

%
%

Let us now discuss the ideas that go into the proofs of Theorems \ref{thm:main} and \ref{thm:alpha-acyclic-no-cycles}. 
%
The proof of Theorem \ref{thm:alpha-acyclic-no-cycles} relies on a useful characterization of $\alpha$-acyclic hypergraphs given in \cite{BFMY} (and stated here as Theorem \ref{thm:HG-alpha-acyclic-S}). This characterization asserts that a hypergraph is $\alpha$-acyclic if and only if it does not contain two certain types of {\em obstructions}. 
These two types of obstructions generalize in different ways the notion of an induced cycle from graphs to hypergraphs.
The main difficulty in the proof of Theorem \ref{thm:alpha-acyclic-no-cycles} lies in translating these ``hypergraph obstructions" into ``digraph obstructions", i.e. recognizing the digraph structures whose reachability hypergraphs correspond to these obstructions. 

We now move on to discuss the proof of Theorem \ref{thm:main}.
With Theorem \ref{thm:alpha-acyclic-no-cycles} at hand, it is natural to first try and show that $\textsc{hom-cnt}_{C_k}$ is hard for all $k \geq 6$. This is indeed accomplished in the following \nolinebreak lemma. 

\begin{lemma}\label{thm:hardness_cycles}
	Assuming Conjecture \ref{conj:triangle_detection}, $\textsc{hom-cnt}_{C_k}$ is hard for every $k \geq 6$.
\end{lemma}

The proof of Lemma \ref{thm:hardness_cycles} works by reducing triangle detection (in general graphs) to the problem of counting colorful $C_k$-homomorphisms (in bounded-degeneracy graphs), which is in turn reduced to \nolinebreak $\textsc{hom-cnt}_{C_k}$.  
	
The remaining ingredient in the proof of Theorem \ref{thm:main} is the following lemma, which states that if $\textsc{hom-cnt}_H$ is easy then so is $\textsc{hom-cnt}_{H'}$ for every induced subgraph $H'$ of $H$. 
It is easy to see that the combination of Theorem \ref{thm:alpha-acyclic-no-cycles} and Lemmas \ref{thm:hardness_cycles} and \ref{thm:hom_cnt_hereditary} implies Theorem \ref{thm:main}. 


\begin{lemma}\label{thm:hom_cnt_hereditary}
	If $\textsc{hom-cnt}_H$ is easy, then $\textsc{hom-cnt}_{H'}$ is easy for every induced subgraph $H'$ of $H$. 
\end{lemma}

The proof of Lemma \ref{thm:hom_cnt_hereditary} uses an innovative application of a powerful technique developed recently in several works on homomorphism-counting, see \cite{CDM,CM}. At the heart of this technique is the observation that for (non-isomorphic) graphs $H_1,\dots,H_k$ and non-zero constants $c_1,\dots,c_k$, the problem of computing the linear combination $c_1\hom(H_1,\cdot) + \dots + c_k \hom(H_k,\cdot)$ is as hard as computing $\hom(H_i,\cdot)$ for every $1 \leq i \leq k$. 
The proof of this fact (appearing in \cite{CDM}) uses tensor products of graphs and a result of Erd\H{o}s, Lov\'{a}sz and Spencer \cite{ELS} (stated here as Lemma \ref{lem:hom_matrix_invertible}) regarding linear independence of homomorphism counts.
For completeness, we present this proof (adapted to the setting of input graphs of bounded degeneracy and stated here as Lemma \ref{thm:hom_linear_combination}) in the appendix. The use of linear combinations of homomorphism-counts plays a crucial role in many of the reductions presented in this paper. 

Our proof of Lemma \ref{thm:hom_cnt_hereditary} proceeds by showing that for every graph $G$, one can (efficiently) construct a graph $G'$ such that $\hom(H,G')$ equals to a linear combination of the homomorphism counts of all induced subgraphs of $H$ in $G$. 
Unlike Lemmas \ref{thm:inj_computation} and \ref{thm:ind_computation}, which are similar (both in their statements and their proofs) to results obtained in \cite{CDM}, Lemma \ref{thm:hom_cnt_hereditary} is (to the best of our knowledge) a new application of the homomorphism-linear-combination technique. 

In the next two sections we take advantage of known relations between the subgraph counts $\hom,\inj,\ind$ in order to derive analogues of Corollary \ref{cor:charac} for the problems $\textsc{inj-cnt}$ and $\textsc{ind-cnt}$. A similar approach was taken in \cite{CDM}.

\subsection{From homomorphisms to copies}\label{subsec:inj}
	In this section we obtain a characterization of graphs $H$ for which $\textsc{inj-cnt}_{H}$ is easy,
	thus resolving Problem \ref{prob:BPS}.
	To state this characterization, we first need to introduce the notion of a {\em quotient graph}. For a graph $H$ and a partition $P = \{U_1,\dots,U_k\}$ of $V(H)$, the {\em quotient graph} $H/P$ is the graph on $P$ in which, for every $1 \leq i,j \leq k$, there is an edge\footnote{Note that if some $U_i$ is not an independent set in $H$, then the vertex $U_i$ has a loop in $H/P$. Such partitions $P$ can be safely ignored in all of our arguments, since our input graphs $G$ are always assumed to be simple, and hence $\hom(H/P,G) = 0$ if $H/P$ has a loop.} 
	between $U_i$ and $U_j$ if and only if there are $u_i \in U_i, u_j \in U_j$ such that $\{u_i,u_j\} \in E(H)$. The following is our main result for $\textsc{inj-cnt}_{H}$. 
	
	

	\begin{theorem}\label{thm:inj_corollary}
		For a graph $H$, $\textsc{inj-cnt}_{H}$ is easy if and only if 
		every quotient graph of $H$ is induced $C_k$-free for all $k \geq 6$. 
		The ``only if" part is under Conjecture \ref{conj:triangle_detection}. 
	\end{theorem}

	One can use Theorem \ref{thm:inj_corollary} to easily obtain the (finite) list of minimal obstructions (forbidden induced subgraphs) for $\textsc{inj-cnt}_{H}$ being easy. However, since this is somewhat lengthy, we shall not give the \nolinebreak details. 

	The reason for the appearance of quotient graphs in Theorem \ref{thm:inj_corollary} is that they can be used to relate homomorphism counts to injective homomorphism counts. Indeed, it is well-known and easy to show (see \cite[Section 5.2.3]{Lovasz}) that 
	\begin{equation}\label{eq:hom-inj}
	\hom(H,G) = \sum_{P}{\inj(H/P,G)},
	\end{equation}
	where $P$ runs over all partitions of $V(H)$. Equation \eqref{eq:hom-inj} can be thought of as a relation over the poset of all partitions of $V(H)$. As is well-known \cite[Section 5.2.3]{Lovasz}, one can invert this relation using M\"{o}bius inversion, thus obtaining the following:
	\begin{equation}\label{eq:hom-inj_inverted}
	\inj(H,G) = \sum_{P} {\mu_{\text{part}}(P) \cdot \hom(H/P,G)}.
	\end{equation}
	Here, $\mu_{\text{part}}$ is the M\"{o}bius function of the partition poset. 
	
	
	Equation \eqref{eq:hom-inj_inverted} expresses $\inj(H,G)$ as a linear combination of $\hom(H/P,G)$ (where $P$ runs over all partitions of $V(H)$). 
	As mentioned above, this means that computing $\inj(H,G)$ is exactly as hard as computing $\hom(H/P,G)$ for all $P$. Thus we have the following:
	\begin{lemma}\label{thm:inj_computation}
		Let $H$ be a graph. Then $\textsc{inj-cnt}_{H}$ is easy if and only if $\textsc{hom-cnt}_{H/P}$ is easy for every partition $P$ of $V(H)$.  
	\end{lemma}
	A similar result has appeared in \cite{CDM}. 
	It is easy to see that Lemma \ref{thm:inj_computation} and Corollary \ref{cor:charac} together imply Theorem \ref{thm:inj_corollary}. 
	A corollary of Lemma \ref{thm:inj_computation} is that $\textsc{inj-cnt}_{H}$ is at least as hard as $\textsc{hom-cnt}_{H}$.

	
	\subsection{From copies to induced copies}\label{subsec:ind}
	In this subsection we apply the results of the previous two subsections in order to obtain a characterization of the graphs $H$ for which $\textsc{ind-cnt}_{H}$ is easy. Here we shall also take the further (simple) step of identifying the minimal obstructions for $\textsc{ind-cnt}_{H}$ being easy. This is done in Theorem \ref{thm:ind_corollary}.
	Here and throughout the paper, a {\em supergraph} of $H$ is any graph on $V(H)$ of which $H$ is a subgraph. 


	\begin{theorem}\label{thm:ind_corollary}
		For a graph $H$,  $\textsc{ind-cnt}_{H}$ is easy if and only if
		$H$ does not contain as an induced subgraph any (not necessarily induced) spanning subgraph of $C_6$. The ``only if" part is under Conjecture \ref{conj:triangle_detection}.  
	\end{theorem}

	Theorem \ref{thm:ind_corollary} implies, for example, that  $\textsc{ind-cnt}_{H}$ is hard if $H$ contains independent set of size $6$, an induced path of length $5$, or an induced matching of size $3$. In the positive direction, we can conclude, for example, that $\textsc{ind-cnt}_{H}$ is easy whenever $H$ has no independent set of size $3$ (i.e., is the complement of a triangle-free graph).

	Theorem \ref{thm:ind_corollary} is derived from a result analogous to Lemma \ref{thm:inj_computation}, this time relating the problems $\textsc{ind-cnt}$ and $\textsc{inj-cnt}$. Again, this relation has been previously used in \cite{CDM}. It is well-known and easy to see that the following holds for every pair of graphs $G,H$. 
	\begin{equation}\label{eq:inj-ind}
	\inj(H,G) = \sum_{E \subseteq \binom{V(H)}{2} \setminus E(H)}{\ind(H \cup E,G)}.	
	\end{equation}
	Here, $H \cup E$ is the graph obtained from $H$ by adding to it all edges in $E$. Hence, $H \cup E$ runs over all supergraphs of $H$.  
	The equation \eqref{eq:inj-ind} can be thought of as a relation over the boolean poset of all subsets of $\binom{V(H)}{2} \setminus E(H)$. Just like \eqref{eq:hom-inj}, it is well-known that this relation can be inverted using M\"{o}bius inversion (which in this case boils down to the inclusion-exclusion principle, see \cite[Section 5.2.3]{Lovasz}). The resulting inverted relation is:
	\begin{equation}\label{eq:inj-ind_inverted}
	\ind(H,G) = \sum_{E \subseteq \binom{V(H)}{2} \setminus E(H)}{(-1)^{|E|} \cdot \inj(H \cup E,G)}.
	\end{equation}
	

	Equation \eqref{eq:inj-ind_inverted} shows that computing $\inj(H',G)$ for every supergraph $H'$ of $H$ is sufficient for computing $\ind(H,G)$. The following lemma states that it is also necessary. The lemma also gives a (seemingly weaker but in fact equivalent) necessary and sufficient condition in terms of computing $\hom(H',G)$ for every supergraph $H'$ of $H$. 
	
	\begin{lemma}\label{thm:ind_computation}
		For every graph $H$, the following are equivalent.
		\begin{enumerate}
			\item $\textsc{ind-cnt}_{H}$ is easy. 
			\item $\textsc{inj-cnt}_{H'}$ is easy for every supergraph $H'$ of $H$.
			\item $\textsc{hom-cnt}_{H'}$ is easy for every supergraph $H'$ of $H$.
		\end{enumerate}
	\end{lemma}

	Given the analogy between \eqref{eq:hom-inj_inverted} and \eqref{eq:inj-ind_inverted}, one could hope for a reduction between $\textsc{ind-cnt}$ and $\textsc{inj-cnt}$ that does not go through $\textsc{hom-cnt}$. We are not aware of such a reduction, however. Instead, we again exploit the homomorphism-linear-combination framework in order to reduce $\textsc{hom-cnt}_{H'}$ (for all supergraphs $H'$ of $H$) to $\textsc{ind-cnt}_H$. We then complete the picture by proving that solving $\textsc{hom-cnt}_{H'}$ for all supergraphs $H'$ of $H$ allows one to solve $\textsc{inj-cnt}_{H'}$ for all such $H'$, which in turn allows one to solve $\textsc{ind-cnt}_H$ using \eqref{eq:inj-ind_inverted}. This step (i.e. proving the implication $3 \Rightarrow 2$) requires Lemma \ref{thm:hom_cnt_hereditary}.
	
	A corollary of Lemma \ref{thm:ind_computation} is that $\textsc{ind-cnt}_{H}$ is at least as hard as $\textsc{inj-cnt}_{H}$ (which itself is at least as hard as $\textsc{hom-cnt}_{H}$).
	At the end of Section \ref{sec:ind}, we give the simple derivation of Theorem \ref{thm:ind_corollary} from Lemma \ref{thm:ind_computation} and Corollary \ref{cor:charac}.

	We conclude by noting that the reductions used to prove Lemmas \ref{thm:hom_cnt_hereditary}, \ref{thm:inj_computation} and \ref{thm:ind_computation} are more robust than is stated in those lemmas. 
	Namely, these reductions pertain not only to algorithms running in time $\tilde{O}(n)$, 
	but to larger running times as well.
	See Lemmas \ref{thm:hom_cnt_hereditary_general}, \ref{thm:inj_computation_general} and \ref{thm:ind_computation_general} for the general statements. 

\subsection{Counting homomorphisms in linear time in general graphs}\label{subsec:general_graphs}


In this section we obtain a characterization of the graphs $H$ such that $\hom(H,\cdot)$ can be computed in (near-)linear time in {\em general} (i.e., not necessarily degenerate) input graphs. Dalmau and Jonsson \cite{DJ} have shown that the complexity of counting $H$-homomorphisms (in general graphs) is essentially controled by the tree-width $\text{tw}(H)$ of $H$: it had been previously shown (see \cite[Proposition 7]{FG}) that $\hom(H,\cdot)$ can be computed in time $O(n^{\text{tw}(H) + 1})$ in $n$-vertex graphs; and conversely, Dalmau and Jonsson \cite{DJ} have shown that there is a function $f : \mathbb{N} \rightarrow \mathbb{N}$ with $f(t) \rightarrow \infty$ (as $t \rightarrow \infty$), such that $\hom(H,\cdot)$ cannot be computed in time $O(n^{f(\text{tw}(H))})$ in $n$-vertex graphs (under the assumption that $\text{FPT}$ does not equal $\#\text{W}[1]$). Here we obtain a sharper\footnote{
Comparing Theorem \ref{thm:charac_general_graphs} to the result of \cite{DJ}, we note that a special case of the latter also uses a reduction from the triangle-counting problem; it shows that, assuming Conjecture \ref{conj:triangle_detection}, $\hom(H,\cdot)$ cannot be computed in time $O(|V(G)|)$ whenever $H$ contains the $3 \times 3$ grid as a minor. Evidently, this condition does not capture all non-forest graphs $H$ (cf. Theorem \ref{thm:charac_general_graphs}).
} 
result for the special case of linear runtime; we show that $\hom(H,\cdot)$ can be computed in linear time {\em if and only if} $\text{tw}(H) = 1$, namely $H$ is a forest. 
\begin{theorem}\label{thm:charac_general_graphs}
	For a graph $H$, $\hom(H,G)$ can be computed in time $\tilde{O}(|V(G)| + |E(G)|)$ if and only if $H$ is a forest. The ``only if" part is under Conjecture \ref{conj:triangle_detection}.  
\end{theorem}

	\paragraph{Some Open Questions}
	By a series of works~\cite{BPS,BPS_SODA21}, including the present paper, 
	we have gained a deep understanding of what kind of patterns can be 
	counted in near-linear time in degenerate graphs. However, from a 
	theoretical standpoint, it would be interesting to explore beyond linear-time algorithms. Specifically, we pose the following question:
	{\em Can we characterize patterns that are countable in quadratic \nolinebreak time?}
	
	A more general question is to determine to which extent the DAG treewidth $\tau_1(H)$ controls the runtime of counting $H$-homomorphisms in bounded-degeneracy graphs. 
	There exist small graphs $H$ for which $\textsc{hom-cnt}_H$ can be solved in time $n^{k}$ for $k < \tau_1(H)$, meaning that the exponent does not always equal $\tau_1(H)$. Still, it could be the case that counting $H$-homomorphisms (in bounded-degeneracy graphs) requires time $n^{f(\tau_1(H))}$, for a function $f(t) : \mathbb{N} \rightarrow \mathbb{N}$ which tends to infinity with $t$; namely, the runtime exponent grows with $\tau_1(H)$. 
	It would be very interesting to determine whether this is indeed the case\footnote{As mentioned above, a result of this type for counting homomorphisms in general graphs is known \cite{DJ}; in this setting, the appropriate parameter is treewidth (and not DAG treewidth).}.
	This question was also considered in the very recent paper of Bressan and Roth \cite{Bressan_Roth}. While the question for $\textsc{hom-cnt}_H$ remains open, \cite{Bressan_Roth} obtained nearly tight bounds for the problems $\textsc{inj-cnt}_H$ and $\textsc{ind-cnt}_H$; for these problems, the parameters governing the runtime exponent are respectively the induced matching number of $H$ (for $\textsc{inj-cnt}_H$) and the independence number of $H$ (for $\textsc{ind-cnt}_H$). These results are incomparable to ours, since we obtain a precise characterization for the case of linear runtime (see Theorems \ref{thm:inj_corollary} and \ref{thm:ind_corollary}), while the results of \cite{Bressan_Roth} are about the asymptotics of the runtime exponent.

	\paragraph{Note}
	The main results presented in this paper have been obtained independently by two groups: one consisting of the first and fourth authors and N. Pashanasangi, and the other consisting of the second, third and fifth authors. In particular, in a preliminary version of the work, Bera et al.~\cite{BPS_SODA21} effectively proved the results stated in Theorems \ref{thm:main} and  \ref{thm:alpha-acyclic-no-cycles}. In their presentation, they use the notion of DAG treewidth instead of the language of $\alpha$-acyclicity.  
	Since the results obtained by the two groups were very similar, we decided to prepare a unified paper. 
	
	\paragraph{Paper organization}
	The rest of the paper is organized as follows. Section \ref{sec:charac} contains the proof of Theorem \ref{thm:alpha-acyclic-no-cycles}. 
	Section \ref{sec:cycles} is devoted to proving Lemma \ref{thm:hardness_cycles}. The proofs of Lemmas \ref{thm:hom_cnt_hereditary}, \ref{thm:inj_computation} and \ref{thm:ind_computation} appear in Sections \ref{sec:hereditary}, \ref{sec:inj} and \ref{sec:ind}, respectively, with Section \ref{sec:ind} also containing the proof of Theorem \ref{thm:ind_corollary}. Finally, in Section \ref{sec:general_graphs} we prove Theorem \ref{thm:charac_general_graphs}.
	
%
	
	\section{A Characterization of $\alpha$-acyclic Graphs: Proof of Theorem \ref{thm:alpha-acyclic-no-cycles}}\label{sec:charac}
	%
	%
	
	The key step in the proof of Theorem \ref{thm:alpha-acyclic-no-cycles} is the following lemma. Recall that for a vertex $u$ in a digraph, $R(u)$ denotes the set of all vertices reachable from $u$. 
	
	\begin{lemma}\label{lem:x_i-in-two-consecutive-T_i-has-cycle}
		Let $\Hdir$ be an acyclic orientation of an undirected graph $H$. Let $k \geq 3$, and assume that there exist distinct vertices 
		$u_0, \dots, u_{k-1},x_0,\dots,x_{k-1} \in V(\Hdir)$ such that for all $0 \leq i \leq k-1$, we have that $x_i \in R(u_{i-1}) \cap R(u_i)$ and $x_i \notin R(u_j)$ for all $j \neq i, i-1$ (with indices taken modulo $k$). Then, $H$ contains an induced copy of a cycle $C_{\ell}$ for some $\ell \geq 6$.
	\end{lemma}
	
	\begin{proof}
		We say that a $2k$-tuple $(v_0, \dots, v_{k-1}, y_0, \dots, y_{k-1})$ of distinct vertices of $\Hdir$ is {\em good} if for every $0 \leq i \leq k-1$,
		$y_i \in R(v_i) \cap R(v_{i-1})$ and $y_i \notin R(v_j)$ for all $j \neq i, i-1$ (with indices taken modulo $k$). In other words, $(v_0, \dots, v_{k-1}, y_0, \dots, y_{k-1})$ is good if and only if for every $0 \leq i \leq k-1$, there are directed paths from $v_i$ to $y_i$ and $y_{i+1}$, and there is no directed path from $v_i$ to $y_j$ for any $j \neq i, i+1$.
		By assumption, the tuple $(u_0, \dots, u_{k-1}, x_0, \dots, x_{k-1})$ is good, implying that the set of good $2k$-tuples is non-empty. 
		Observe that for a good $2k$-tuple $(v_0, \dots, v_{k-1}, y_0, \dots, y_{k-1})$ and $0 \leq i \leq k-1$, it holds that $R(v_i) \cap R(v_{i-1}) \cap \{y_0,\dots,y_{k-1}\} = \{y_i\}$ (here we use the assumption that $k \geq 3$). This implies that $R(v_i) \cap R(v_{i-1}) \cap \{v_0,\dots,v_{k-1}\} = \emptyset$, because if $v_j \in R(v_i) \cap R(v_{i-1})$ (for some $0 \leq j \leq k-1$), then $y_j,y_{j+1} \in R(v_j) \subseteq R(v_i) \cap R(v_{i-1})$, which we just ruled out. Similarly, the definition of a good $2k$-tuple implies that if $y_i,y_{i+1} \in R(v_j)$ (for some $0 \leq i,j \leq k-1$), then $j = i$ (again, we are using here the assumption that $k \geq 3$). This implies that there are no $0 \leq i,j \leq k-1$ such that $y_i,y_{i-1} \in R(y_j)$, since otherwise we would have $y_i,y_{i-1} \in R(y_j) \subseteq R(v_j) \cap R(v_{j-1})$, which is impossible since we cannot have both $j = i$ and $j-1 = i$. 
		We have thus established the following fact, which will be used several \nolinebreak times. 
		\begin{fact}\label{obs:alpha-acyclic_good_2k-tuple}
			Let $(v_0,\dots,v_{k-1},y_0,\dots,y_{k-1})$ be a good tuple,
			let $z \in \{v_0,\dots,v_{k-1},y_0,\dots,y_{k-1}\}$ and let $0 \leq i \leq k-1$. If $z \in R(v_i) \cap R(v_{i-1})$ then $z = y_i$, and if $y_i,y_{i+1} \in R(z)$ then $z = v_i$. 
		\end{fact}

		For vertices $a,b \in V(H) = V(\Hdir)$, denote by $\didist(a,b)$ the length of a shortest directed path from $a$ to $b$ (in $\Hdir$). 
		Now, fix a good $2k$-tuple $M = (v_0, \dots, v_{k-1}, y_0, \dots, y_{k-1})$ which minimizes the sum
		\begin{equation}\label{eq:sum-paths}
		\sum_{i=0}^{k-1} \left( \didist(v_i, y_i)+\didist(v_{i-1}, y_i) \right).
		\end{equation}
		
		Let us now fix specific shortest (directed) paths $P(v_i, y_j)$ from $v_i$ to $y_j$ for $0 \leq i \leq k-1$ and $j=i, i+1$ (as always, indices are taken modulo $k$). We will denote by $P\{v_i,y_j\}$ the underlying undirected path of $P(v_i,y_j)$.
		Let $C$ be the (undirected) closed walk obtained by concatenating \nolinebreak the \nolinebreak paths 
		\begin{equation}\label{eq:paths}
		P\{y_0, v_0\},P\{v_0, y_1\},P\{y_1, v_1\},\dots,P\{v_{k-2}, y_{k-1}\},P\{y_{k-1}, v_{k-1}\},P\{v_{k-1}, y_0\}.
		\end{equation}
		We now show that $C$ is a simple cycle. We will then show that $C$ is induced. While the proof idea is rather simple, the details are somewhat lengthy.
		
		Two paths appearing consecutively (in a cyclic manner) in \eqref{eq:paths} will be called consecutive. 
		In other words, the pairs of consecutive paths are $(P\{y_i,v_i\},P\{v_i,y_{i+1}\})$ and $(P\{v_i,y_{i+1}\},P\{y_{i+1},v_{i+1}\})$ (for $0 \leq i \leq k-1$).
		If, by contradiction, $C$ is not a simple cycle, then either there is a pair of non-consecutive paths which intersect, or a pair of consecutive paths which intersect outside of the endpoint they share. We now rule out each of these possibilities. 
		
		\noindent \textbf{Case 1:} We consider the intersection of $P(v_i, y_i)$ and $P(v_i, y_{i+1})$ for some $0 \leq i \leq k-1$. Assume that there exists a vertex $z \neq v_i$ such that $z \in P(v_i, y_i) \cap P(v_i, y_{i+1})$. 
		Then $y_i,y_{i+1} \in R(z)$. By Fact \ref{obs:alpha-acyclic_good_2k-tuple}, we have $z \notin \{v_0,\dots,v_{k-1},y_0,\dots,y_{k-1}\}$. 
		Now, replacing $v_i$ with $z$ in the tuple $M$, we get a new $2k$-tuple of distinct vertices which is also good. Indeed, there are paths from $z$ to $y_i, y_{i+1}$, and there is no path from $z$ to $y_j$ for all $j \neq i, i+1$, as otherwise we would have a path from $v_i$ to $y_j$ (via $z$), contradicting the assumption. Since $z \neq v_i$, the two new paths we get by replacing $v_i$ with $z$ are strictly shorter than the two original ones, which contradicts the minimality of \eqref{eq:sum-paths}.
		
		\noindent \textbf{Case 2:} We consider the intersection of $P(v_{i-1}, y_i)$ and $P(v_i, y_i)$ for some $0 \leq i \leq k-1$. Assume that there exists a vertex $z \neq y_i$ such that $z \in P(v_{i-1}, y_i) \cap P(v_i, y_i)$. Then $z \in R(v_i) \cap R(v_{i-1})$. By Fact \ref{obs:alpha-acyclic_good_2k-tuple}, $z \notin \{v_0,\dots,v_{k-1},y_0,\dots,y_{k-1}\}$.
		Now, replacing $y_i$ with $z$ in the tuple $M$, we get a new $2k$-tuple of distinct vertices which is also good. Indeed, there are paths from $v_{i-1}, v_i$ to $z$, and there is no path from $v_j$ to $z$ for any $j \neq i, i-1$, as otherwise we would have a path from $v_j$ to $y_i$ (via $z$), which contradicts the assumption. Since $z \neq y_i$, the two new paths we get by replacing $y_i$ with $z$ are strictly shorter than the two original ones, which contradicts the minimality of \eqref{eq:sum-paths}.
		
		\noindent \textbf{Case 3:} We consider the intersection of $P(v_i, y_s)$ and $P(v_j, y_t)$ for some $i \neq j$ and $s \neq t$ (where $s \in \{i, i+1\}$ and $t \in \{j, j+1\}$). Assume that there exists a vertex $z$ such that $z \in P(v_i, y_s) \cap P(v_j, y_t)$. In this case, as there are paths from $v_i$ and $v_j$ to $y_s, y_t$ (via $z$), we must have $\{s,t\} = \{i,i+1\} = \{j,j+1\}$, which implies that $i = j$ (as $k \geq 3$), in contradiction to the assumption that $i \neq j$. 
		
		From Cases 1-3 it follows that the closed walk $C$ is indeed a simple cycle in $H$.
		We now prove that $C$ is an induced cycle. We first observe that for a path $P(v_i, y_j)$ with $j \in {i, i+1}$ and two vertices $z_1, z_2 \in P(v_i, y_j)$ such that $z_2$ comes after $z_1$ along the path, we cannot have the edge $(z_2, z_1)$ as $\Hdir$ is acyclic. In addition, we cannot have the edge $(z_1, z_2)$, unless it is an edge of the path, as this would contradict the fact that $P(v_i, y_j)$ is a shortest path from $v_i$ to $y_j$. Thus, the (undirected) path $P\{v_i,y_j\}$ is induced. 
		We conclude that if $C$ has a chord, then it must connect two distinct paths among the paths in \eqref{eq:paths}. We now rule out the existence of such a chord by analyzing several cases.
		
		\noindent \textbf{Case 1:} Consider an edge between the two paths $P(v_i, y_i)$ and $P(v_i, y_{i+1})$ for some $0 \leq i \leq k-1$. 
		We assume, without loss of generality, that there is an edge $(z_1, z_2) \in E(\Hdir)$ such that $z_1 \in P(v_i, y_i)$, $z_2 \in P(v_i, y_{i+1})$, and $z_1, z_2 \neq v_i$. Then $y_i \in R(z_1)$ and $y_{i+1} \in R(z_2) \subseteq R(z_1)$. By Fact \ref{obs:alpha-acyclic_good_2k-tuple}, $z_1 \notin \{v_0,\dots,v_{k-1},y_0,\dots,y_{k-1}\}$. 
		Now, replacing $v_i$ with $z_1$ in the tuple $M$, we get a new $2k$-tuple of distinct vertices which is also good. Indeed, there are paths from $z_1$ to $y_i, y_{i+1}$, and there is no path from $z_1$ to $y_j$ for all $j \neq i, i+1$, as otherwise we would have a path from $v_i$ to $y_j$ (via $z_1$), which is impossible. Since $z_1, z_2 \neq v_i$, we have 
		$\didist(z_1,y_i) + \didist(z_1,y_{i+1}) < \didist(v_i,y_i) + \didist(v_i,y_{i+1})$, which contradicts the minimality of \eqref{eq:sum-paths}.
		
		\noindent \textbf{Case 2:} Consider an edge between the two paths $P(v_{i-1}, y_i)$ and $P(v_i, y_i)$ for some $0 \leq i \leq k-1$. 
		We assume, without loss of generality, that there is an edge $(z_1, z_2) \in E(\Hdir)$ such that $z_1 \in P(v_{i-1}, y_i)$, $z_2 \in P(v_i, y_i)$, and $z_1, z_2 \neq y_i$.
		Then $z_2 \in R(v_i) \cap R(z_1) \subseteq R(v_i) \cap R(v_{i-1})$. By Fact \ref{obs:alpha-acyclic_good_2k-tuple}, $z_2 \notin \{v_0,\dots,v_{k-1},y_0,\dots,y_{k-1}\}$.
		Now, replacing $y_i$ with $z_2$ in the tuple $M$, we get a new $2k$-tuple of distinct vertices which is also good. Indeed, there are paths from $v_{i-1}, v_i$ to $z_2$, and there is no path from $v_j$ to $z_2$ for all $j \neq i, i-1$, as otherwise we would have a path from $v_j$ to $y_i$ (via $z_2$), which is impossible. Since $z_1, z_2 \neq y_i$, we have 
		$\didist(v_i,z_2) + \didist(v_{i-1},z_2) < \didist(v_i,y_i) + \didist(v_{i-1},y_i)$, which contradicts the minimality of (\ref{eq:sum-paths}).
		
		\noindent \textbf{Case 3:} Consider an edge between the two paths $P(v_i, y_s)$ and $P(v_j, y_t)$ for some $0 \leq i \neq j \leq k-1$ and $s \neq t$ (where $s \in \{i, i+1\}$ and $t \in \{j, j+1\}$). We assume, without loss of generality, that there is an edge $(z_1, z_2) \in E(\Hdir)$ such that $z_1 \in P(v_i, y_s)$ and $z_2 \in P(v_j, y_t)$. In this case, as there is a path from $v_i$ to $y_t$ (via $z_1, z_2$), we must have $\{s,t\} = \{i,i+1\}$. In addition, as $t \in \{j, j+1\}$, we either have that $s=i, t=i+1$, which implies $t=j$ (as $i \neq j$); or that $s=i+1, t=i$, which implies $t=j+1$ \nolinebreak (as \nolinebreak $i \neq j$).
		
		We first consider the case when $s=i, t=i+1$ and $j=i+1$. 
		In other words, we are considering the situation where there is an edge $(z_1, z_2) \in E(\Hdir)$ from $z_1 \in P(v_i,y_i)$ to $z_2 \in P(v_{i+1},y_{i+1})$.  
		Note that $y_i \in R(z_1)$ and $y_{i+1} \in R(z_2) \subseteq R(z_1)$. By Fact \ref{obs:alpha-acyclic_good_2k-tuple}, either $z_1 = v_i$ or $z_1 \notin \{v_0,\dots,v_{k-1},y_0,\dots,y_{k-1}\}$. Similarly, note that $z_2 \in R(z_1) \cap R(v_{i+1}) \subseteq R(v_i) \cap R(v_{i+1})$, so by Fact \ref{obs:alpha-acyclic_good_2k-tuple} either $z_2 = y_{i+1}$ or $z_2 \notin \{v_0,\dots,v_{k-1},y_0,\dots,y_{k-1}\}$. Furthermore, we cannot have both $z_1 = v_i$ and $z_2 = y_{i+1}$, because then $z_1,z_2$ are both contained in the path $P(v_i,y_{i+1})$, and we have already ruled out the possibility of such a chord.  
		
		Now, replacing $v_i$ with $z_1$ and $y_{i+1}$ with $z_2$ in the tuple $M$, we get a new $2k$-tuple of distinct vertices which is also good. Indeed, there are paths from $z_1$ to $y_i, z_2$, and there is no path from $z_1$ to $y_j$ for all $j \neq i, i+1$, as otherwise we would have a path from $v_i$ to $y_j$ (via $z_1$), which contradicts the assumption. Similarly, there are paths from $z_1, v_{i+1}$ to $z_2$, and there is no path from $v_j$ to $z_2$ for all $j \neq i, i+1$, as otherwise we would have a path from $v_j$ to $y_{i+1}$ (via $z_2$), which contradicts the assumption. 
		Since either $z_1 \neq v_i$ or $z_2 \neq y_{i+1}$, we have 
		$\didist(z_1,y_i) + \didist(z_1,z_2) + \didist(v_{i+1},z_2) < \didist(v_i,y_i) + \didist(v_i,y_{i+1}) + \didist(v_{i+1},y_{i+1})$, 
		which contradicts the minimality of \eqref{eq:sum-paths}. 
		
		We now consider the (symmetrical) case when $s=i+1, t=i$ and $j=i-1$. In other words, we are considering the situation where there is an edge $(z_1, z_2) \in E(\Hdir)$ from $z_1 \in P(v_i,y_{i+1})$ to $z_2 \in P(v_{i-1},y_i)$. Note that $y_{i+1} \in R(z_1)$ and $y_i \in R(z_2) \subseteq R(z_1)$. By Fact \ref{obs:alpha-acyclic_good_2k-tuple}, either $z_1 = v_i$ or $z_1 \notin \{v_0,\dots,v_{k-1},y_0,\dots,y_{k-1}\}$. Similarly, note that $z_2 \in R(z_1) \cap R(v_{i-1}) \subseteq R(v_i) \cap R(v_{i-1})$, so by Fact \ref{obs:alpha-acyclic_good_2k-tuple} either $z_2 = y_i$ or $z_2 \notin \{v_0,\dots,v_{k-1},y_0,\dots,y_{k-1}\}$. Furthermore, we cannot have both $z_1 = v_i$ and $z_2 = y_i$, because then $z_1,z_2$ are both contained in the path $P(v_i,y_i)$, and we have already ruled out the possibility of such a chord.  
		
		Now, replacing $v_i$ with $z_1$ and $y_i$ with $z_2$ in the tuple $M$, we get a new $2k$-tuple of distinct vertices which is also good. Indeed, there are paths from $z_1$ to $z_2, y_{i+1}$, and there is no path from $z_1$ to $y_j$ for all $j \neq i, i+1$, as otherwise we would have a path from $v_i$ to $y_j$ (via $z_1$), which contradicts the assumption. Similarly, there are paths from $z_1, v_{i-1}$ to $z_2$, and there is no path from $v_j$ to $z_2$ for all $j \neq i, i-1$, as otherwise we would have a path from $v_j$ to $y_i$ (via $z_2$), which contradicts the assumption. 
		Since either $z_1 \neq v_i$ or $z_2 \neq y_i$, we have 
		$\didist(z_1,y_{i+1}) + \didist(z_1,z_2) + \didist(v_{i-1},z_2) < \didist(v_i,y_{i+1}) + \didist(v_i,y_i) + \didist(v_{i-1},y_i)$, 
		which contradicts the minimality of \eqref{eq:sum-paths}. 
		
		Cases 1-3 imply that $C$ is an induced cycle. The length of $C$ is evidently at least $2k \geq 6$. This completes the proof of the lemma. 
	\end{proof}
	
	
	To prove Theorem \ref{thm:alpha-acyclic-no-cycles}, we combine Lemma \ref{lem:x_i-in-two-consecutive-T_i-has-cycle} with the following theorem, which gives a structural characterization of $\alpha$-acyclic hypergraphs. 
	The proof of this theorem can be found in \cite{BFMMUY,BFMY}. 
	
	\begin{theorem}[\cite{BFMMUY,BFMY}]\label{thm:HG-alpha-acyclic-S}
		A hypergraph $F$ is $\alpha$-acyclic if and only if for every $k \geq 3$, there is {\em no} \linebreak $S = \{x_0, x_1, \dots, x_{k-1}\} \subseteq V(F)$ such that one of the following conditions holds:
		\begin{enumerate}
			\item For every $0 \leq i \leq k-1$ there exists $e \in E(F)$ such that $e \cap S = \{x_i, x_{i+1}\}$, and there is no $e \in E(F)$ with $|e \cap S| \geq 2$ such that $e \cap S \neq \{x_i, x_{i+1}\}$ for all $0 \leq i \leq k-1$. (All indices are taken modulo $k$.) 
			
			\item For every $0 \leq i \leq k-1$ there exists $e \in E(F)$ such that $e \cap S = S \setminus \{x_i\}$, and there is no $e \in E(F)$ such that $S \subseteq e$. 
		\end{enumerate}
	\end{theorem}
	
	%
	\noindent
	We are now ready to prove Theorem \ref{thm:alpha-acyclic-no-cycles}. 
	
	\begin{proof}[Proof of Theorem \ref{thm:alpha-acyclic-no-cycles}] 
		We first prove the ``only if" part of the theorem, or, more precisely, the contrapositive of this statement. 
		Suppose that $H$ contains an induced copy of $C_{\ell}$ for some $\ell \geq 6$, and let $C = (c_0, c_1, \dots, c_{\ell-1})$ be such a cycle.  
		We now orient the edges of $H$ as follows:
		the edges of $C$ are oriented alternatingly along the cycle. More precisely, for an even $0 \leq i \leq \ell-1$, we orient the edge $\{c_i, c_{i+1}\} \in E(H)$ from $c_i$ to $c_{i+1}$, and for an odd $0 \leq i \leq \ell-1$, we orient the edge $\{c_i, c_{i+1}\} \in E(H)$ from $c_{i+1}$ to $c_i$ (with indices taken modulo $\ell$). 
		Now, orient all edges between $C$ and $V(H) \setminus C$ from $C$ to $V(H) \setminus C$, and orient all edges in $V(H) \setminus C$ arbitrarily while avoiding the creation of a directed cycle.
		We denote the resulting orientation of $H$ by $\Hdir$. One can easily verify that $\Hdir$ is acyclic. Indeed, by our construction, the induced subgraphs $\Hdir[C]$ and $\Hdir[V(H) \setminus C]$ do not contain directed cycles, and all the edges between $C$ and $V(H) \setminus C$ go in the same direction. 
		Let 
		$S = \{c_j \, | \, j \text{ is odd}\}$, noting that $|S| \geq \lfloor \ell/2 \rfloor \geq 3$.
		Observe that if $\ell$ is even, then the vertices $c_i$ for even $0 \leq i \leq \ell-1$ are sources, and $R(c_i) \cap S = \{c_{i-1},c_{i+1}\}$ for each such $i$. If $\ell$ is odd, then the sources are the vertices $c_i$ with even $2 \leq i \leq \ell-1$; also, $R(c_i) \cap S = \{ c_{i-1},c_{i+1} \}$ for each even $2 \leq i \leq \ell-3$, and $R(c_{\ell-1}) \cap S = \{c_{\ell-2},c_1\}$, as $R(c_{\ell-1}) \cap V(C) = \{c_{\ell-2},c_0,c_1\}$.
		Observe also that for every source vertex $u$ of $\Hdir$ which is not in $C$, we have that $R(u) \cap S = \emptyset$. So we see that in both cases ($\ell$ even or odd), the situation in Item 1 of Theorem \ref{thm:HG-alpha-acyclic-S} holds with respect to the above-chosen $S$. 
		Hence, by Theorem \ref{thm:HG-alpha-acyclic-S}, $\Hdir$ is not $\alpha$-acyclic. Thus, $H$ has an acyclic orientation which is not $\alpha$-acyclic, as required. 
		
		We now establish the ``if" part of the theorem. Again, we will prove the contrapositive. Suppose that there exists an acyclic orientation $\Hdir$ of $H$ which is not $\alpha$-acyclic. 
		Let $\FHdir$ be the hypergraph as in Definition \ref{def:alpha-acyclic-dir}. Then $\FHdir$ is not $\alpha$-acyclic. 
		Hence, by Theorem \ref{thm:HG-alpha-acyclic-S}, there exists $S = \{x_0, x_1, \dots, x_{k-1}\} \subseteq V(\FHdir)$ with $k \geq 3$ such that (at least) one of the conditions 1-2 in that theorem holds with respect to $S$. Our goal is to show that this implies the existence of an induced $\ell$-cycle in $H$ for some $\ell \geq 6$. 
		
		Assume first that $S$ satisfies Condition 1 of Theorem \ref{thm:HG-alpha-acyclic-S}. For each $0 \leq i \leq k-1$, let $e_i \in E(\FHdir)$ be such that $e_i \cap S = \{x_i, x_{i+1}\}$. By the definition of $\FHdir$, for each $0 \leq i \leq k-1$ there is a source $u_i$ of $\Hdir$ such that $e_i = R(u_i)$.  
		So for every $0 \leq i \leq k-1$ we have that $x_i \in R(u_{i-1}) \cap R(u_i)$ and $x_i \notin R(u_j)$ for all $j \neq i, i-1$. 
		Moreover, $u_0, \dots, u_{k-1},x_0,\dots,x_{k-1}$ are pairwise-distinct because $x_0,\dots,x_{k-1}$ are pairwise-distinct, $u_0,\dots,u_{k-1}$ are pairwise-distinct, and $u_0, \dots, u_{k-1}$ are sources of $\Hdir$ while $x_0,\dots,x_{k-1}$ are not. 
		Therefore, we can apply Lemma \ref{lem:x_i-in-two-consecutive-T_i-has-cycle} and get that $H$ contains an induced copy of $C_{\ell}$ for some $\ell \geq 6$, as required.
		
		Now assume that $S$ satisfies Condition 2 of Theorem \ref{thm:HG-alpha-acyclic-S}. 
		For each $0 \leq i \leq k-1$, let $e_i \in E(\FHdir)$ be such that $e_i \cap S = S \setminus \{x_i\}$. By the definition of $\FHdir$, for each $0 \leq i \leq k-1$ there is a source $u_i$ of $\Hdir$ such that $e_i = R(u_i)$.  
		Considering only $R(u_0), R(u_1), R(u_2)$, and setting $y_0 := x_1, y_1 := x_2, y_2 := x_0$, we have that $y_0 \in R(u_2) \cap R(u_0)$ but $y_0 \notin R(u_1)$, $y_1 \in R(u_0) \cap R(u_1)$ but $y_1 \notin R(u_2)$, and $y_2 \in R(u_1) \cap R(u_2)$ but $y_2 \notin R(u_0)$. Furthermore, just like in the previous case, $u_0,u_1,u_2,y_0,y_1,y_2$ are pairwise-distinct. Therefore, we can again apply Lemma \ref{lem:x_i-in-two-consecutive-T_i-has-cycle} and get that $H$ contains an induced copy of $C_{\ell}$ for some $\ell \geq 6$, as required.	
	\end{proof}

	\section{Hardness Results for $\textsc{HOM-CNT}_{C_k}$: Proof of Lemma \ref{thm:hardness_cycles}}\label{sec:cycles}
	
	The proof of Lemma \ref{thm:hardness_cycles} uses so-called colorful subgraph counts. Let $H$ be a graph on $h$ vertices. 
	For a graph $G$ and a coloring $c : V(G) \rightarrow [h]$ of its vertices, a homomorphism $\varphi: H \rightarrow G$ is called {\em colorful} if every color appears exactly once in the image of $\varphi$, namely, if $\{c(\varphi(v)) : v \in V(H)\} = [h]$. \linebreak(Note that the number of colors is always the number of vertices in $H$.)
	Evidently, a colorful homomorphism must be injective. We denote by $\text{col-inj}(H,G,c)$ the number of colorful homomorphisms from $H$ to $G$, and by $\textsc{col-inj-cnt}_H$ the problem of computing the number of colorful homomorphisms of $H$ in an $h$-vertex-colored input graph. 
	We use the known fact that $\textsc{col-inj-cnt}_H$ reduces to $\textsc{hom-cnt}_H$ (see e.g. \cite{Meeks}). For completeness, we include a proof. 
	\begin{lemma}\label{lem:colorful_hom_count}
		If one can solve $\textsc{hom-cnt}_H$ in time $f(n)$ in $n$-vertex graphs then one can also solve $\textsc{col-inj-cnt}_H$ in time $O(f(n))$ in $n$-vertex graphs. Moreover, the reduction preserves degeneracy.
	\end{lemma}
	\begin{proof}
		Suppose we are given a graph $G$ and a coloring $c : V(G) \rightarrow [h]$, and wish to compute the number of colorful $H$-homomorphisms in $G$. For each set $I \subseteq [h]$, let $V_I$ be the set of vertices of $G$ whose color belongs to $I$. By definition, a homomorphism $\varphi : H \rightarrow G$ is colorful if and only if 
		$c(\text{Image}(\varphi)) = [h]$.
		For every 
		$I \subseteq [h]$, the number of homomorphisms $\varphi : V(H) \rightarrow V(G)$ such that $c(\text{Image}(\varphi)) \subseteq I$ is exactly $\hom(H,G[V_I])$. So by inclusion-exclusion, we have
		$$
		\text{col-inj}(H,G,c) = 
		\sum_{I \subseteq [h]}{(-1)^{|I|} \cdot \hom(H,G[V_I])}.
		$$
		So by computing $\hom(H,G[V_I])$ for each $I \subseteq [h]$, we can recover $\text{col-inj}(H,G,c)$.
		Evidently, each of the graphs $G[V_I]$ has at most $n$ vertices and its degeneracy is at most that of $G$. This completes the proof. 
	\end{proof}
	
	We will also need the fact that detecting a colorful triangle in a $3$-vertex-colored graph is at least as hard as uncolored triangle detection.
	
	\begin{lemma}\label{lem:color_coding}
		If one can detect a colorful triangle in a $3$-vertex-colored graph with $m$ edges in time $f(m)$, then one can also detect a triangle in an uncolored graph with $m$ edges in time $\tilde{O}(f(m))$. 
	\end{lemma}  
	\begin{proof}
		This follows from the color-coding method, see \cite{AYZ_ColorCoding}. Suppose we are given an $n$-vertex graph $G$ and need to decide if $G$ is triangle free. Color the vertices of $G$ randomly with three colors. If $G$ contains a triangle, then the resulting $3$-vertex-colored graph will contain a colorful triangle with constant positive probability. Hence, if there is a $f(n)$-time algorithm for colorful triangle detection, then there is also a $O(f(n))$-time randomized algorithm for (uncolored) triangle detection which has constant success probability. Using color-coding \cite{AYZ_ColorCoding}, this can be derandomized at the cost of a logarithmic factor, thus giving the desired $\tilde{O}(f(n))$-time triangle detection algorithm. 
	\end{proof}
	
	\noindent
	The following lemma constitutes the main step in the proof of Lemma \ref{thm:hardness_cycles}.
	
	\begin{lemma}\label{lem:cycles_main}
		Let $k \geq 4$. For every graph $G$ and coloring $c : V(G) \rightarrow [3]$, one can construct in time $O(|V(G)| + |E(G)|)$ a graph $G'$ with 
		$|V(G')|,|E(G')| = O(|V(G)| + |E(G)|)$ 
		and a coloring $c' : V(G') \rightarrow [k]$, such that $\mathrm{col\text{-}inj}(C_k,G',c') = \frac{k}{3} \cdot \mathrm{col\text{-}inj}(C_3,G,c)$. Moreover, if $k \geq 6$ then $G'$ is $2$-degenerate. 
	\end{lemma}

	\begin{proof}
		Fix numbers 
		$\ell_{1,2},\ell_{1,3},\ell_{2,3} \in \left\{ \lfloor k/3 \rfloor - 1, \lceil k/3 \rceil - 1 \right\}$ such that $\ell_{1,2} + \ell_{1,3} + \ell_{2,3} = k - 3$. 
		Fix a partition $\{4,\dots,k\} = I_{1,2} \cup I_{1,3} \cup I_{2,3}$ where $|I_{i,j}| = \ell_{i,j}$ for each pair $1 \leq i < j \leq 3$.
		We define $G'$ and $c'$ as follows. For each edge $\{x,y\} \in E(G)$, if $c(x) \neq c(y)$, say $i = c(x)$ and $j = c(y)$, $1 \leq i \neq j \leq 3$, then replace the edge $\{x,y\}$ with a path of length $\ell_{i,j}+1$ (between $x$ and $y$), and  color the $\ell_{i,j}$ internal vertices of this path with the $\ell_{i,j}$ colors in $I_{i,j}$. If $c(x) = c(y)$ then simply remove the edge $\{x,y\}$. To complete the definition of $c'$, we set $c'(v) = c(v) \in [3]$ for every $v \in V(G) \subseteq V(G')$. 
		The resulting graph is $G'$ and the resulting coloring is $c'$. Note that the vertices in $V(G)$ are colored with $1,2,3$, while the vertices in $V(G') \setminus V(G)$ are colored with $4,\dots,k$. 
		We have
		$|V(G')| \leq |V(G)| + (\lceil k/3 \rceil - 1) \cdot |E(G)| = O(|V(G)| + |E(G)|)$ and $|E(G')| \leq \lceil k/3 \rceil \cdot |E(G)| = O(|V(G)| + |E(G)|)$. Moreover, if $k \geq 6$ then $G'$ is $2$-degenerate.
		Indeed, observe that the original vertices of $G$ form an independent set in $G'$, because each edge of $G$ is either removed or replaced by a path of length at least
        $\lfloor k/3 \rfloor \geq 2$. 
        Also, all vertices in $V(G') \setminus V(G)$ have degree $2$ in $G'$. Upon removing these vertices, we are left with an empty graph; hence $G'$ is $2$-degenerate, as claimed. 
		
		Let us now consider the colorful $k$-cycles in $G'$. It is easy to see that every cycle $C$ in $G'$ corresponds to a cycle in $G$, in the sense that there is a cycle $x_1,x_2,\dots,x_r,x_1$ in $G$ such that $C$ consists of $x_1,\dots,x_r$ and of the paths which replaced the edges $\{x_1,x_2\},\dots,\{x_{r-1},x_r\},\{x_r,x_1\}$. Now, if $C$ is a colorful $k$-cycle, then it must be the case that $r = 3$ and that the triangle $x_1,x_2,x_3$ is colorful. Indeed, if $C$ is colorful then $x_1,\dots,x_r$ must have distinct colors, which is impossible if $r \geq 4$ (as the vertices of $G$ are colored with $1,2,3$). In the other direction, the definition of $G'$ and $c'$ implies that every colorful triangle in $G$ gives rise to a colorful $k$-cycle in $G'$. It is now easy to see that $\mathrm{col\text{-}inj}(C_k,G',c') = \frac{k}{3} \cdot \mathrm{col\text{-}inj}(C_3,G,c)$, as required  
	\end{proof}
	
	\noindent
	Lemma \ref{thm:hardness_cycles} now follows easily from the above lemmas:
	
	\begin{proof}[Proof of Lemma \ref{thm:hardness_cycles}]
		Let $k \geq 6$, and suppose that $\textsc{hom-cnt}_{C_k}$ can be solved in time $\tilde{O}(n)$ in $2$-degenerate $n$-vertex graphs. By Lemma \ref{lem:colorful_hom_count}, one can also solve $\textsc{col-inj-cnt}_{C_k}$ in time $\tilde{O}(n)$ in $2$-degenerate  $k$-vertex-colored $n$-vertex graphs. One can then use the construction given by Lemma \ref{lem:cycles_main} to obtain an algorithm for colorful triangle detection in ($3$-vertex-colored) $m$-edge graphs, running in time $\tilde{O}(m)$. Using Lemma \ref{lem:color_coding}, one obtains an algorithm for (uncolored) triangle detection, which runs in time $\tilde{O}(m)$ on $m$-edge graphs. The existence of such an algorithm is ruled out by Conjecture \ref{conj:triangle_detection}. 
	\end{proof}

	Using the same argument, one obtains the following statement regarding the hardness of counting $C_k$-homomorphisms in general (i.e., not necessarily degenerate) graphs.
	
	\begin{lemma}\label{lem:cycles_general_graphs}
		Let $k \geq 4$. 
		Assuming Conjecture \ref{conj:triangle_detection}, $\textsc{hom-cnt}_{C_k}$ cannot be solved in time $\tilde{O}(|V(G)| + \nolinebreak |E(G)|)$. 
	\end{lemma}
	
	\section{Solvability of $\textsc{HOM-CNT}_H$ is Hereditary: Proof of Lemma \ref{thm:hom_cnt_hereditary}}\label{sec:hereditary}
	In this section we prove Lemma \ref{thm:hom_cnt_hereditary}. We will actually prove the following more general statement.
	\begin{lemma}\label{thm:hom_cnt_hereditary_general}
		For every graph $H$ there is $k = k(H)$ such that the following holds. For every graph $G$ there are graphs $G_1,\dots,G_k$, computable in time $O(|V(G)| + |E(G)|)$, such that  $|V(G_i)| = O(|V(G)|)$ and $|E(G_i)| = O(|E(G)|)$ for every $i = 1,\dots,k$, and such that
		knowing $\hom(H,G_1),\dots,\hom(H,G_k)$ allows one to compute $\hom(H',G)$ for all induced subgraphs $H'$ of $H$ in time $O(1)$. Furthermore, if $G$ is \linebreak$O(1)$-degenerate, then so are $G_1,\dots,G_k$. 
	\end{lemma}
	\noindent
	It is easy to see that Lemma \ref{thm:hom_cnt_hereditary_general} implies Lemma \ref{thm:hom_cnt_hereditary}. 
	
	We now state an important lemma concerned with computing {\em linear combinations of homomorphism counts}. This is implicit in \cite{CDM}.
	As mentioned in the introduction, several of the reductions presented in this paper --- including the proof of Lemma \ref{thm:hom_cnt_hereditary_general} --- rely on it.   
%
	

\begin{lemma}\label{thm:hom_linear_combination}
	Let $H_1,\dots,H_k$ be pairwise non-isomorphic graphs and let $c_1,\dots,c_k$ be non-zero constants. For every graph $G$ there are graphs $G_1,\dots,G_k$, computable in time $O(|V(G)| + |E(G)|)$, such that  $|V(G_i)| = O(|V(G)|)$ and $|E(G_i)| = O(|E(G)|)$ for every $i = 1,\dots,k$, and such that
	knowing $b_j := c_1 \cdot \hom(H_1,G_j) + \dots + c_k \cdot \hom(H_k,G_j)$ for every $j = 1,\dots,k$ allows one to compute $\hom(H_1,G),\dots,\hom(H_k,G)$ in time $O(1)$. Furthermore, if $G$ is $O(1)$-degenerate, then so are $G_1,\dots,G_k$. 
\end{lemma}


	\noindent
	For completeness, we give a proof of Lemma \ref{thm:hom_linear_combination} in the appendix.
	
	In the proof of Lemma \ref{thm:hom_cnt_hereditary_general}, it will be convenient to consider the empty graph $K_0$, and to define 
	$\hom(K_0,G) = 1$ for every graph $G$.

	\begin{proof}[Proof of Lemma \ref{thm:hom_cnt_hereditary_general}]
		Let $H$ be a graph on $h$ vertices.
		For a graph $F$, we denote by $F + K_h$ the graph obtained from $F$ by adding a clique of size $h$ and connecting it to $V(F)$ with a complete bipartite graph. 
		We start by showing that for every graph $F$, it holds that
		\begin{equation}\label{eq:hom_cnt_graph+clique}
		\hom(H,F + K_h) = 
		\sum_{U \subseteq V(H)} \hom(H[U], F) \cdot \hom(H[V(H) \setminus U], K_h).
		\end{equation}
		To see that \eqref{eq:hom_cnt_graph+clique} holds, let us assign to each function $\varphi : V(H) \rightarrow V(F + K_h)$ the set 
		\linebreak$U = U(\varphi) := \varphi^{-1}(V(F)) \subseteq V(H)$. Our definition of the graph $F + K_h$ guarantees that a function $\varphi : V(H) \rightarrow V(F + K_h)$ is a homomorphism if and only if $\varphi|_{U(\varphi)}$ is a homomorphism from $H[U]$ to $F$ and $\varphi|_{V(H) \setminus U(\varphi)}$ is a homomorphism from $H[V(H) \setminus U(\varphi)]$ to $K_h$. By summing over all possible values of $U$, we get \eqref{eq:hom_cnt_graph+clique}. Note that $U(\varphi)$ may be empty (in case $\text{Im}(\varphi) \subseteq K_h$); in this case $\hom(H[U],F) = 1$ (as $H[U]$ is the empty graph).
		
		Let $H_1, H_2, \dots, H_k$ be an enumeration of all induced subgraphs of $H$ (including the empty one), up to isomorphism (that is, $H_1,\dots,H_k$ are pairwise non-isomorphic). For each $1 \leq i \leq k$, set
		$$
		c_i := \sum_{\substack{U \subseteq V(H): \\ H[U] \cong H_i}} \hom(H[V(H) \setminus U], K_h).
		$$
		Note that $c_1,\dots,c_k$ depend only on $H$.
		With this notation, we can rewrite \eqref{eq:hom_cnt_graph+clique} as follows:
		\begin{equation}\label{eq:hom_cnt_graph+clique_1}
		\hom(H,F + K_h) = 
		\sum_{i = 1}^{k}{c_i \cdot \hom(H_i,F)}.
		\end{equation}
		Note that for each $1 \leq i \leq k$ we have $c_i > 0$, since there is some $U \subseteq V(H)$ for which $H[U] \cong H_i$ (by our choice of $H_1,\dots,H_k$), and for this $U$ it clearly holds that $\hom(H[V(H) \setminus U], K_h) > 0$. 
		
		Now let $G$ be a graph. 
		Apply Lemma \ref{thm:hom_linear_combination} (to the graph $G$), and let  $G'_1,\dots,G'_k$ be the graphs given by that lemma. For each $1 \leq i \leq k$, set $G_i := G'_i + K_h$, noting that $|V(G_i)| = |V(G'_i)| + h = \linebreak O(|V(G)| + |E(G)|)$ and $|E(G_i)| = |E(G'_i)| + |V(G'_i)| \cdot h + \binom{h}{2} = O(|V(G)| + |E(G)|)$, where in both cases the last equality is guaranteed by Lemma \ref{thm:hom_linear_combination}. Furthermore, if $G$ is $O(1)$-degenerate then so is $G'_i$ for every $1 \leq i \leq k$ (by Lemma \ref{thm:hom_linear_combination}), and hence so is $G_i$ for every $1 \leq i \leq k$ (indeed, if a graph $F$ is $\kappa$-degenerate then $F + K_h$ is $(\kappa+h)$-degenerate). Crucially, observe that knowing $\hom(H,G_1),\dots,\hom(H,G_k)$ allows one to compute $c_1 \cdot \hom(H_1,G'_j) + \dots + c_k \cdot \hom(H_k,G'_j)$ for every $1 \leq j \leq k$ (by \eqref{eq:hom_cnt_graph+clique_1}), which in turn allows one to compute $\hom(H_1,G),\dots,\hom(H_k,G)$ in time $O(1)$ (by our choice of $G'_1,\dots,G'_k$ via Lemma \ref{thm:hom_linear_combination}). This completes the proof, as every induced subgraph of $H$ is isomorphic to one of the graphs $H_1,\dots,H_k$.
	\end{proof}
%
	
	\section{Counting Copies: Proof of Lemma \ref{thm:inj_computation}}\label{sec:inj}
	We prove the following more general lemma, from which Lemma \ref{thm:inj_computation} will easily follow.
	\begin{lemma}\label{thm:inj_computation_general}
		For every graph $H$ there is $k = k(H)$ such that the following holds. For every graph $G$ there are graphs $G_1,\dots,G_k$, computable in time $O(|V(G)| + |E(G)|)$, such that $|V(G_i)| = O(|V(G)|)$ and $|E(G_i)| = O(|E(G)|)$ for every $i = 1,\dots,k$, and such that
		knowing $\inj(H,G_1),\dots,\inj(H,G_k)$ allows one to compute $\hom(H/P,G)$ for all partitions $P$ of $V(H)$ in time $O(1)$. Furthermore, if $G$ is $O(1)$-degenerate, then so are $G_1,\dots,G_k$. 
	\end{lemma} 

	In the proof, we will need the value (or rather, the sign) of the M\"{o}bius function of the partition poset (which plays a role in \eqref{eq:hom-inj_inverted}). This M\"{o}bius function is given by the Frucht--Rota-–Sch\"{u}tzenberger formula (see, e.g., \cite[Chapter 25]{vLW}): for a partition $P$ of the vertex-set of $H$, we have:
	\begin{equation}\label{eq:Frucht_Rota_Schutzenberger}
	\mu_{\text{part}}(P) = (-1)^{|V(H)| - |P|} \cdot \prod_{U \in P}{(|U|-1)!}.
	\end{equation}
	\begin{proof}
		[Proof of Lemma \ref{thm:inj_computation_general}]
		The lemma follows easily by combining \eqref{eq:hom-inj_inverted} and Lemma \ref{thm:hom_linear_combination}. 
		Let $H_1,\dots,H_k$ be an enumeration of all quotient graphs of $H$, up to isomorphism. That is, $H_1,\dots,H_k$ are pairwise non-isomorphic and $\{H_1,\dots,H_k\} = 
		\{H/P : P \in \mathcal{P}(H)\},$
		where, as before, $\mathcal{P}(H)$ denotes the set of all partitions of $V(H)$.
		Let $G$ be a graph. 
		By using \eqref{eq:hom-inj_inverted} and ``combining like terms", we get that
		\begin{equation}\label{eq:hom-inj_tensor-product}
		\begin{split}
		\inj(H,G) &= 
		\sum_{P \in \mathcal{P}(H)}{\mu_{\text{part}}(P) \cdot \hom(H/P,G)} = 
		\sum_{i=1}^{k}{
			\left( \sum_{\substack{P \in \mathcal{P}(H) : \\ H/P \cong H_i}}{\mu_{\text{part}}(P)} \right) \cdot \hom(H_i,G)}.
		\end{split}
		\end{equation}
		From \eqref{eq:Frucht_Rota_Schutzenberger} we know that $\mu_{\text{part}}(P) \neq 0$ for all $P \in \mathcal{P}(H)$, and that the sign of $\mu_{\text{part}}(P)$ depends only on the number of parts in $P$. In particular, if $P,Q \in \mathcal{P}(H)$ are such that $H/P \cong H/Q$, then $\mu_{\text{part}}(P)$ and $\mu_{\text{part}}(Q)$ have the same sign. Now, setting 
		$$
		c_i := \sum_{\substack{P \in \mathcal{P}(H) : \\ H/P \cong H_i}}{\mu_{\text{part}}(P)},
		$$
		we see that $c_i$ is non-zero, since the summands in the above sum cannot cancel each other. With this notation, \eqref{eq:hom-inj_tensor-product} becomes 
		\begin{equation}\label{eq:hom-inj_tensor-product_1}
		\inj(H,G) = \sum_{i = 1}^{k}{c_i \cdot \hom(H_i,G)}.
		\end{equation}
		Now the lemma immediately follows from \eqref{eq:hom-inj_tensor-product_1} and Lemma \ref{thm:hom_linear_combination}, as every quotient graph of $H$ is isomorphic to one of the graphs $H_1,\dots,H_k$. 
	\end{proof}
	\begin{proof}[Proof of Lemma \ref{thm:inj_computation}]
		The ``if" part of the lemma follows from \eqref{eq:hom-inj_inverted}, and the ``only-if" part of the lemma follows from Lemma \ref{thm:inj_computation_general}. 
	\end{proof}

	\section{Counting Induced Copies: Proof of Theorem \ref{thm:ind_corollary} and Lemma \ref{thm:ind_computation}}\label{sec:ind}
	Before establishing Lemma \ref{thm:ind_computation}, we again prove a more general statement.
	\begin{lemma}\label{thm:ind_computation_general}
		For every graph $H$ there is $k = k(H)$ such that the following holds. For every graph $G$ there are graphs $G_1,\dots,G_k$, computable in time $O(|V(G)| + |E(G)|)$, such that  $|V(G_i)| = O(|V(G)|)$ and $|E(G_i)| = O(|E(G)|)$ for every $i = 1,\dots,k$, and such that the following holds:
		\begin{enumerate}
			\item Knowing 
			$\ind(H,G_i)$ for every $1 \leq i \leq k$ 
			allows one to compute $\hom(H',G)$ for all supergraphs $H'$ of $H$ in time $O(1)$.
			\item Knowing $\hom(H',G_i)$ for every supergraph $H'$ of $H$ and every $1 \leq i \leq k$, allows one to compute $\inj(H',G)$ for all supergraphs $H'$ of $H$ in time $O(1)$.
		\end{enumerate}
		Furthermore, if $G$ is $O(1)$-degenerate, then so are $G_1,\dots,G_k$.
	\end{lemma}
	\begin{proof}
		We will show that there are $G_1,\dots,G_k$ which satisfy the assertion of each of the items 1-2 separately. One can then take the union of these two families of graphs to obtain the desired graphs $G_1,\dots,G_k$ for which both items hold. 
		
		Starting with Item 2, we begin by proving the following preliminary claim: 
		for every supergraph $H'$ of $H$ and every partition $P$ of $V(H)$, there is a supergraph $H''$ of $H$ such that $H'/P$ is an induced subgraph of $H''$. Indeed, fixing $H'$ and $P$ as above, we define $H''$ as follows: for each $\{U,V\} \in E(H'/P)$ (so $U,V \in P$), add to $H$ all edges between $U$ and $V$; the resulting graph is $H''$. It is easy to see that $H'/P$ is indeed an induced subgraph of $H''$, as required. 
		
		Now let $G$ be a graph.  
		By combining the above claim with Lemma \ref{thm:hom_cnt_hereditary_general} (which we apply to all supergraphs $H'$ of $H$ at once), we see that there are graphs $G_1,\dots,G_k$ with $|V(G_i)| = O(|V(G)|)$ and $|E(G_i)| = O(|E(G)|)$ for every $i = 1,\dots,k$, such that knowing $\hom(H',G_i)$ for every supergraph $H'$ of $H$ and every $1 \leq i \leq k$, allows one to compute $\hom(H'/P,G)$ for all supergraphs $H'$ of $H$ and all partitions $P$ of $V(H)$ in time $O(1)$. 
		With this information at hand, one can use \eqref{eq:hom-inj} to compute $\inj(H',G)$ for all supergraphs $H'$ of $H$ in time $O(1)$, as required. 
	
		We now move on to establish Item 1. For convenience, we denote by $\mathcal{E}$ the set of all subsets of  $\binom{V(H)}{2} \setminus E(H)$, and by $\mathcal{P}$ the set of all partitions of $V(H)$. 
		We start by observing that for every graph \nolinebreak $F$, 
		\begin{equation}\label{eq:hom-ind_tensor-product}
		\begin{split}
		\ind(H,F) &= 
		\sum_{E \in \mathcal{E}}{(-1)^{|E|} \cdot \inj(H \cup E,F)} 
		\\ &=
		\sum_{E \in \mathcal{E}} \;\sum_{P \in \mathcal{P}}
		{
			(-1)^{|E|} \cdot \mu_{\text{part}}(P) \cdot \hom((H \cup E)/P,F)
		},
		\end{split}
		\end{equation}
		where the first equality uses \eqref{eq:inj-ind_inverted} and the second equality uses \eqref{eq:hom-inj_inverted}.
		For an (unlabeled) graph $H'$, put
		\begin{equation}\label{eq:hom-ind_coeff}
		c_{H'} := \sum_{\substack{E \in \mathcal{E}, \; P \in \mathcal{P}: \\ (H \cup E)/P \cong H'}}{(-1)^{|E|} \cdot \mu_{\text{part}}(P)}.
		\end{equation}
		With this notation, \eqref{eq:hom-ind_tensor-product} becomes 
		\begin{equation*}
		\ind(H,F) = \sum_{H' : \; c_{H'} \neq 0}{c_{H'} \cdot \hom(H',F)}.
		\end{equation*}	
		Let $H_1,\dots,H_k$ be an enumeration of all graphs $H'$ such that $c_{H'} \neq 0$, up to isomorphism. In other words, $H_1,\dots,H_k$ are pairwise non-isomorphic and 
		$\{H_1,\dots,H_k\} = 
		\{H' : c_{H'} \neq 0\}$. For each $1 \leq i \leq k$, we put $c_i := c_{H_i}$. 
		Finally, we can write
		\begin{equation}\label{eq:hom-ind_tensor-product_1}
		\ind(H,F) = \sum_{i = 1}^{k}{c_i \cdot \hom(H_i,F)}.
		\end{equation}	
		Crucially, observe that every supergraph of $H$ is isomorphic to one of the graphs $H_1,\dots,H_k$. To see this, let $H'$ be a (representative of the isomorphism class of a) supergraph of $H$. Clearly, if $E \in \mathcal{E}$ and $P \in \mathcal{P}$ are such that $(H \cup E)/P \cong H'$, then $P$ is the partition of $V(H)$ into singletons, and $|E| = |E(H')| - |E(H)|$. Hence, all summands on the right-hand side of \eqref{eq:hom-ind_coeff} have the same sign, implying that $c_{H'} \neq 0$. This in turn implies that $H'$ is isomorphic to one of the graphs $H_1,\dots,H_k$, as required. 
		Now Item 1 of the lemma immediately follows from \eqref{eq:hom-ind_tensor-product_1} and Lemma \ref{thm:hom_linear_combination}.  
	\end{proof}
	\begin{proof}[Proof of Lemma \ref{thm:ind_computation}]
	The implication $2 \Rightarrow 1$ follows from \eqref{eq:inj-ind_inverted}, the implication $3 \Rightarrow 2$ follows from Item 2 of Lemma \ref{thm:ind_computation_general}, and the implication $1 \Rightarrow 3$ follows from Item 1 of Lemma \ref{thm:ind_computation_general}. 
	\end{proof}

	\begin{proof}[Proof of Theorem \ref{thm:ind_corollary}]
		By putting together Lemma \ref{thm:ind_computation} and Corollary \ref{cor:charac}, we see that $\textsc{ind-cnt}_H$ is easy if and only if every supergraph of $H$ is induced $C_k$-free for every $k \geq 6$ (with the ``only if" part requiring Conjecture \ref{conj:triangle_detection}). 
		Thus, to prove the theorem, it remains to observe that the following are equivalent:
		\begin{enumerate}
			\item Every supergraph of $H$ is induced $C_k$-free for every $k \geq 6$.
			\item $H$ does not contain as an induced subgraph any (not necessarily induced) spanning subgraph of $C_6$.
		\end{enumerate}
		The implication $1 \Rightarrow 2$ is immediate. For the reverse implication, observe that if some supergraph $H'$ of $H$ contains an induced copy of $C_k$ (for some $k \geq 6$) then by adding an appropriate chord to this $C_k$, we obtain a supergraph $H''$ of $H$ which contains an induced copy of $C_6$. But then $H$ contains an induced copy of some spanning subgraph of $C_6$. This shows that $2 \Rightarrow 1$. 
	\end{proof}

	\section{Proof of Theorem \ref{thm:charac_general_graphs}}\label{sec:general_graphs}

	The ``if" part of Theorem \ref{thm:charac_general_graphs} follows from the following proposition. 
	
	\begin{proposition}\label{prop:hom_count_forest}
		Let $H$ be a forest. Then for every graph $G$, $\hom(H,G)$ can be computed in time $\tilde{O}(|V(G)| + |E(G)|)$. 
	\end{proposition}
	
	\begin{proof}
		First, observe that if $H$ is a disconnected graph with connected components $H_1,\dots,H_k$, then $\hom(H,G) = \prod_{i=1}^{k}\hom(H_i,G)$ for every graph $G$. Thus, to prove the proposition, it suffices to consider the case that $H$ is a tree. 
		
		Suppose then that $H$ is a tree. To compute $\hom(H,G)$, we use dynamic programming to compute, for each subtree $H'$ of $H$, $v \in V(H')$ and $x \in V(G)$, the number $N(H',v,x)$ of homomorphisms $\varphi$ from $H'$ to $G$ such that $\varphi(v) = x$. To compute $N(H',v,x)$ for given $H'$ and $v$ (simultaneously for all $x \in V(G)$), we consider two cases according to whether or not $v$ is a leaf of $H'$. Suppose first that $v$ is a leaf of $H'$, let $u$ be the only neighbour of $v$ in $H'$, and set $H'' := H' - v$. Now, for each $x \in V(G)$, compute $N(H',v,x) = \sum_{y : \{x,y\} \in E(G)}{N(H'',u,y)}$ (where the sum is over all neighbours $y$ of $x$). Suppose now that $v$ is not a leaf, let
		$C_1,\dots,C_k$ be the connected components of the forest $H' - v$, and set $H'_i := H'[C_i \cup \{v\}]$. Now, for each $x \in V(G)$, compute $N(H',v,x) = \prod_{i=1}^{k}{N(H'_i,v,x)}$. 
		Thus, in both cases, one can compute $N(H',v,x)$ for all $x \in V(G)$ by relying on counts for smaller trees which have already been computed and stored. Retrieving or updating stored values can be done in time $\tilde{O}(1)$.
	\end{proof}
	\begin{proof}[Proof of Theorem \ref{thm:charac_general_graphs}]
		The ``if" part of the theorem follows from Proposition \ref{prop:hom_count_forest}. The ``only if" part follows from the combination of Lemmas \ref{lem:cycles_general_graphs} and \ref{thm:hom_cnt_hereditary_general} (as any non-forest evidently contains an induced \nolinebreak cycle). 
	\end{proof}
	
	\paragraph{Acknowledgments} 
	We would like to thank an anonymous referee for kindly pointing out that our methods can be used to study homomorphism-counting in general graphs, thus leading to Theorem \ref{thm:charac_general_graphs}. We are also grateful to Marc Roth and Raphael Yuster for informing us of the reduction currently used to prove Lemma \ref{thm:hardness_cycles} (a previous version of the paper used a different, more complicated reduction).

\appendix
\section{Proof of Lemma \ref{thm:hom_linear_combination}}
We start by recalling the definition of the tensor product of graphs. The tensor product of graphs $G_1$ and $G_2$, denoted $G_1 \times G_2$ has vertex set 
$V(G_1 \times G_2) = V(G_1) \times V(G_2)$ and edge-set 
$$E(G_1 \times G_2) = 
\{\{(x_1,x_2), (y_1,y_2)\} \, : \, \{x_1,y_1\} \in E(G_1) \text{ and } \{x_2,y_2\} \in E(G_2)\}.$$
A key property of the tensor product is that the parameter $\hom(H,\cdot)$ is multiplicative with respect to it (for any graph $H$). That is, for every pair of graphs $G_1,G_2$, it holds that
\begin{equation}\label{eq:hom-multiplicative-for-tensor}
\hom(H, G_1 \times G_2) = \hom(H, G_1) \cdot \hom(H, G_2). 
\end{equation}
To see that \eqref{eq:hom-multiplicative-for-tensor} holds, simply observe that for functions $\varphi_i : V(H) \rightarrow V(G_i)$ (where $1 \leq i \leq 2$), the function $v \mapsto (\varphi_1(v),\varphi_2(v))$ is a homomorphism from $H$ to $G_1 \times G_2$ if and only if $\varphi_i$ is a homomorphism from $H$ to $G_i$ for each $1 \leq i \leq 2$. 
In what follows, we will use the following (trivial) observation regarding tensor products and degeneracy.
\begin{observation}\label{obs:tensor_product_degrees}
	Let $F,G$ be graphs. If $G$ is $\kappa$-degenerate, then $F \times G$ is 
	$(v(F) \cdot \kappa)$-degenerate. 
\end{observation}
\begin{proof}
	It is easy to see that for each $x \in V(F)$ and $y \in V(G)$, the degree of $(x,y)$ in $F \times G$ is $d_{F \times G}((x,y)) = d_F(x) \cdot d_G(y) < v(F) \cdot d_G(y)$. It follows that every subgraph of $F \times G$ contains a vertex of degree at most $v(F) \cdot \kappa$ (since $G$ is $\kappa$-degenerate).
\end{proof}

We now state a lemma of Erd\H{o}s, Lov\'{a}sz and Spencer \cite{ELS} (see also Proposition 5.44(b) in \cite{Lovasz}), which will play a crucial role in the proof of Lemma \ref{thm:hom_linear_combination}. 
\begin{lemma}[\cite{Lovasz}]\label{lem:hom_matrix_invertible}
	Let $H_1, \dots, H_k$ be pairwise non-isomorphic graphs, and let 
	$c_1, \dots, c_k \neq 0$ be non-zero constants. Then there exist graphs $F_1, \dots, F_k$ such that the $k \times k$ matrix 
	$M_{i,j} = c_j \cdot \nolinebreak \hom(H_j, F_i)$, $1 \leq i,j \leq k$, is invertible. 
\end{lemma}

\noindent
Finally, we are ready to prove Lemma \ref{thm:hom_linear_combination}. 

\begin{proof}[Proof of Lemma \ref{thm:hom_linear_combination}]
	By Lemma \ref{lem:hom_matrix_invertible}, there are graphs $F_1,\dots,F_k$  such that the $k \times k$ matrix \linebreak $M_{i,j} := c_j \cdot \hom(H_j,F_i)$ $(1 \leq i,j \leq k)$ is invertible. 
    Given an input graph $G$, we set $G_i := F_i \times G$ and $b_i := c_1 \cdot \hom(H_1, G_i) + \dots + c_k \cdot \hom(H_k, G_i)$ for each $1 \leq i \leq k$. Observe that 
	\begin{equation}\label{eq:tensor-product_reduction}
	\begin{split}
	b_i = \sum_{j = 1}^{k}{c_j \cdot \hom(H_j,F_i \times G)} = 
	\sum_{j = 1}^{k}{c_j \cdot \hom(H_j,F_i) \cdot \hom(H_j,G)} = 
	\sum_{j = 1}^{k}{M_{i,j} \cdot \hom(H_j,G)}.
	\end{split}
	\end{equation}
	In the third equality above, we used \eqref{eq:hom-multiplicative-for-tensor}. 
	We will treat \eqref{eq:tensor-product_reduction} ($1 \leq i \leq k$) as a system of linear equations, where $\hom(H_1,G),\dots,\hom(H_k,G)$ are the variables, $M$ is the matrix of the system, and $b_1,\dots,b_k$ are the constant terms. Since $M$ is invertible (as guaranteed by our choice of $F_1,\dots,F_k$), knowing $b_1,\dots,b_k$ indeed enables us to compute $\hom(H_1,G),\dots,\hom(H_k,G)$ in time $O(1)$, as required. 
	To complete the proof, we note that $|V(G_i)| = |V(G)| \cdot |V(F_i)| = O(|V(G)|)$ and $|E(G_i)| \leq |E(G)| \cdot |E(F_i)|^2 = O(|E(G)|)$ for every $1 \leq i \leq k$, and that if $G$ is $O(1)$-degenerate then so are $G_1,\dots,G_k$ by 
	Observation \ref{obs:tensor_product_degrees}.
\end{proof}
\end{document}